%% file: ms.tex
\documentclass[conference,letterpaper]{IEEEtran}
\makeatother
 \pdfoutput=1

\usepackage{booktabs} 
\usepackage{color}

\usepackage{multirow}
\usepackage{array}
\usepackage[tight,footnotesize]{subfigure}
\usepackage[pdftex]{graphicx}
\usepackage[bookmarks=false]{hyperref}
\usepackage{amsfonts}
\usepackage{amsthm}

\newcommand{\spdz}{SPDZ}

\newcommand{\di}{\displaystyle}
\newcommand{\mod}{~\textnormal{mod}~}

\newcommand{\acronym}{SEMBA}
\newtheorem{theorem}{Theorem}

\begin{document}
\title{\acronym: SEcure Multi-Biometric Authentication}

\author{\IEEEauthorblockN{Giulia Droandi, Tommaso Pignata\\ and Mauro Barni}
	\IEEEauthorblockA{Department of Information Engineering and Mathematics\\
		University of Siena, Siena, Italy\\
		Email: \{giulia.droandi,pignata.tommaso\}@gmail.com\\
	barni@diism.unisi.it}
	\and
	\IEEEauthorblockN{Riccardo Lazzeretti}
	\IEEEauthorblockA{Department of Computer, Control, and \\
		Management Engineering "Antonio Ruberti"\\ Sapienza University of Rome, Rome, Italy.\\
		Email: lazzeretti@diag.uniroma1.it}
}

\maketitle

\begin{abstract}
Biometrics security is a dynamic research area spurred by the need to protect personal traits from threats
 like theft, non-authorised  distribution, reuse and so on.
  A widely investigated solution to such threats consists in processing the biometric signals 
  under encryption, to avoid any leakage of information towards non-authorised parties.
In this paper, we propose to leverage on the superior performance of multimodal biometric recognition
 to improve the efficiency of a biometric-based authentication protocol operating on encrypted data under
  the malicious security model. In the proposed protocol, authentication relies on both facial and iris biometrics, 
  whose representation accuracy is specifically tailored to trade-off between recognition accuracy and
   efficiency. From a cryptographic point of view, the protocol relies on SPDZ \cite{damgaard2012multiparty, damgaard2013practical}. 
Experimental results show that the multimodal protocol is faster than  corresponding unimodal protocols achieving the same accuracy. 
\end{abstract}


\section{Introduction }
\input{intro}
\section{Prior works}
\label{sec_stateArt}
\input{protectionBio}
\section{Tools}\label{sec_tools}
In this section, we present the cryptographic  and biometrical tools used in our 
protocol.
\subsection{Cryptographic tools: SPDZ system} 
\label{sec_spdz}
\input{spdz}

\subsection{Biometrics tools}
\label{sec_bio}
\input{Biometrics}

\section{Proposed Protocol} 
\label{sec_implements}
In this section, we describe the details of  \acronym~implementation. We start by presenting the security model (Section \ref{sec_model}), 
then we present the implementation of the iris and face authentication protocols (Sections \ref{sec_iris_Encrypted} and \ref{sec_faceEncry})
 and the multimodal biometric protocol (Section \ref{sec_fusionEncry}).
  Finally, we discuss the security of the \spdz$\,$ protocols in Section \ref{sec_securityDemonstration}.
\subsection{Security model}\label{sec_model}
All the protocols involve a client (the biometric owner) and a server that authenticates the identity of the client. 
Since in our computation we have only two parties, from now on, the \spdz $ \, 
$
protocol is described for $n=2$.
On one hand the client does not want to reveal his biometric templates, on the other hand the server does not want to dislocate its records.
Both client and server can be malicious, i.e. they may be interested in gaining as much information as possible on the other party even by deviating from the protocol. 
Considering that the \spdz$\,$ protocol involving $n$ {{parties}}  is secure up to $n-1$ malicious parties, we assume that only one between the client and the server can act maliciously. 
We underline that if both act maliciously, they obtain no real information about the counterpart.
We also assume that the parties are connected through a secure channel providing privacy against eavesdroppers and any third party that can compromise the transmission.

\input{implem}

\input{tests}

%
  
\section{Conclusions}\label{sec_conclusion}
\input{conclusion}

\bibliographystyle{IEEEtran}
\bibliography{bibliografia} 

\end{document}

%% file: intro.tex

In the digital and increasingly interconnected world we live in, establishing individuals' identity is a pressing need. 
For this reason, in the last decades, we have seen an increasing interest in biometric-based recognition systems. 
Biometric recognition can be split into two main categories: authentication and identification.
In the first scenario, also referred to as \emph{verification},
the user is interested in demonstrating that he/she is who he/she claims to
be, while in the second one, the goal is to determine the identity of the user submitting the biometric
template among those \emph{known by the system}.  
%
%
Usually, in both verification and identification protocols, a single biometric trait
is used to extract a feature vector. The feature vector then is matched with one or 
more templates stored in the system database. 
In this work we focus on an authentication protocol.

More recently, the security of biometric systems has become a very active research area, due to the necessity of impeding newly emerging cybercrimes like identity theft, privacy violation, unauthorized access to sensitive information and so on \cite{SPM_rivista}. 
Protocols allowing to process encrypted biometric signals without decrypting them are among the most widely studied solutions to enhance the security of biometric systems \cite{bringer2013privacy,barni2015privacy}. 
According to such an approach, verification or identification is carried out by the system by relying only on encrypted biometric templates, thus avoiding the risk that sensitive information is leaked during the protocol.

The possibility of processing a comparing encrypted biometric templates relies on a number of cryptographic tools \cite{michael1981rabin,Yao86,paillier1999public,GH11}, broadly referred to as Multi-Party Computation (MPC) \cite{goldreich1998secure}. 
%
Generally speaking, MPC protocols can be classified according to the adopted security model. 
The most common distinction considers protocols which are secure only against \emph{semi-honest}  adversaries, and those which can be proven to be secure also against \emph{malicious adversaries}.
To be specific, according to the \emph{semi-honest} model, all the parties execute the 
protocol without deviating from it, but meanwhile they try to obtain as much
information as possible about the other parties' data. %
Protocols developed in the semi-honest model are very efficient and, for this reason, are adopted in the majority of the works proposed so far \cite{bringer2013privacy,barni2015privacy}.
%
On the contrary, in the \emph{malicious} model, the parties can arbitrarily deviate from 
the protocol in their attempt to get access to sensitive information. %
While security against malicious parties is desirable, in many real world applications, the resulting protocols have a very high complexity and their use in real systems is often impractical.
The great majority of the attempts made so far to devise efficient biometric recognition protocols in the malicious setting, focused on the development and use of innovative and efficient MPC and cryptographic primitives.
A much less investigated approach consists in the adoption of biometric recognition protocols which are better suited to be implemented in a MPC framework. Yet, as highlighted in \cite{barni2015privacy}, working on the signal processing side of the problem may help to reduce significantly the complexity of the resulting MPC  protocol, e.g. by efficiently trading off between accuracy and complexity. %

\subsection{Contribution}
In this work, we follow the above strategy and present \acronym: a SEcure Multi-Biometric Authentication protocol which achieves a better trade-off between efficiency and accuracy with respect to the single modality subsystems composing it. 
This represents a major departure from most works on multimodal biometric systems, in which the availability of multiple biometric modalities is exploited to decrease interclass variability and improve intra-class similarity in the presence of acquisition noise and any other kind of distortion \cite{ross2008introduction
}.
In this framework, the main contributions of the paper are the following:
\textit{i)} we design a multimodal biometric system that combines face and iris templates and that can be easily implemented by relying on Secure Multiparty Computation protocols;
\textit{ii)} we propose a privacy preserving multi-biometric authentication protocol secure against a malicious party. \acronym~ is based on the \spdz~tool \cite{damgaard2012multiparty, damgaard2013practical} and discloses only the final binary decision;
\textit{iii)} we compare our multi-biometric protocol with single biometric protocols, showing that by using a properly simplified representation of the two biometric traits, backed by a rigorous signal processing analysis, the multimodal protocol can reach the same accuracy of the corresponding single-modality systems based on more accurate - and more complicated - representations of iris and face templates, but with significantly lower computational complexity. In particular, \acronym~obtains the same accuracy of the stand
alone iris authentication protocol described in \cite{masek2003recognition}. 
Of course, system designers could also decide to exploit the superior performance allowed by
multimodal authentication to improve authentication accuracy with the same complexity of the single modality protocols.

%% file: protectionBio.tex
In the last years, many cryptographic tools, 
including oblivious transfer \cite{michael1981rabin}, homomorphic encryption 
\cite{paillier1999public, pisa2012somewhat}, and garbled circuits \cite{
Yao86}, have been used for privacy protection of biometrics templates. 
In most works, such tools are used in such a way to achieve security in the semi-honest model. 
Many privacy preserving authentication protocols have been proposed in the literature making use of a wide variety of biometric traits.
Since, in this work, we present a privacy preserving multibiometric authentication protocol based on face and iris, we focus on the state of art relative to those  biometries, then we discuss the few works achieving privacy protection in the malicious model and finally we discuss the main characteristics of multimodal (or \emph{fusion}) biometric systems.

\subsection{Biometric recognition in the semi-honest model}

As pointed in \cite{bringer2013privacy,barni2015privacy}, many prior works on biometric recognition are designed to be secure against semi-honest 
adversaries.
Some examples are \cite{erkin2009privacy,sadeghi2010efficient} for face recognition 
and
\cite{luo2012efficient,bringer2012faster,blanton2011secure} for iris recognition. 
In order to guarantee security, 
the above papers are based on many cryptographic techniques 
such as Pailler additive homomorphic cryptosystem (HE), Oblivious Transfer (OT) or  Garbled 
Circuit (GC). 
For efficiency reasons, such protocol are mainly based on the eigenface \cite{turk1991face} and iriscode \cite{daugman2004} representation of iris and face respectively. More accurate protocols exist, but their privacy preserving implementation has such a high complexity to make them impractical.

\subsection{Biometric recognition in malicious setting }\label{sec_multiChiaro}
There are few works on privacy preserving biometric authentication secure under a malicious model.
Kantarcioglu and  Kardes \cite{kantarcioglu2008privacy}  present a way to implement some 
primitives, specifically the dot product and equality check, in the malicious model, by also analyzing the corresponding computational 
cost. Even if this work is not directly related to biometrics, the proposed solutions can be adapted to them. 
In \cite{abidin2016privacy}, Abidin presents a general framework for biometric authentication that uses a homomorphic encryption scheme to evaluate the distance between two encrypted  biometric templates. 
In his work, Abidin proves security against malicious attacks, but does not provide results about the practical implementation of the protocol. 
In \cite{pathak2013privacy}, Pathak and  Raj present two speech-based authentication protocols.
The first one is an interactive protocol based on Pailler cryptosystem which is secure against a semi-honest adversary, 
the second one is a non-interactive protocol based on BNG \cite{boneh2005evaluating} cryptosystem, 
which allows to perform an arbitrary number of additions and one multiplication between ciphertexts, 
and is secure against malicious attacks. In both cases the output is a probability value and the client checks if such a value is equal to zero or not in the plain domain. From the tests and the analysis reported in \cite{pathak2013privacy} it is clear that the interactive protocol is more efficient than the malicious one.

A large number of approaches \cite{kiraz2006protocol,lindell2008implementing,pinkas2009secure,
	nielsen2012new,lindell2016fast} have been proposed to make Yao's garbled circuit techniques secure in the malicious model through Zero Knowledge proof, cut and choose, or other techniques. Such approaches can also be used for biometric authentication protocols, but their complexity is so high to make them impractical.
Gasti et al. \cite{gasti2016secure} propose a lightweight biometric authentication protocol based on simple garbled circuits and secure against malicious adversaries by relying on an untrusted third party (the cloud). The goal of the protocol is to minimize the amount of computation performed by the biometric owner's device (a smartphone), while also reducing the protocol execution time and without the necessity to rely on cut-and-choose techniques.
In the protocol, the biometric owner acts as circuit constructor, the cloud as circuit evaluator, while the server verifies the correctness of the circuit.
The approach is secure against colluding biometric owner and cloud, but not against colluding server and cloud.

\subsection{Multimodal biometric recognition}
{Given the recent technological advances, novel devices are often equipped with numerous sensors, opening the way to multi-biometric authentication. 
 In  \cite{ross2006handbook} Ross, Nandakumar, and Jain 
present an overview of the possible fusion scenarios and their applications in real life.
For our protocol we choose a \textit{multimodal} system, that
combines information from face and iris.}

Biometric signals are usually processed in four stages.
First a sensor captures the traits of an individual as a raw 
biometric data. Second, raw data are 
processed and a compact representation of the physical 
traits, called \emph{features},
is extracted.
Then the feature template is matched with the
templates stored in a database. Finally the matching score is used to 
determine an identity or to validate a claimed identity.
Information can be merged at any time during a multi-biometric recognition protocol.
{For further information about this topic readers may refer to \cite{ross2006handbook}.
The choice of fusion depends on the intended application, its specific 
characteristics and the  multiparty computation tool chosen to guarantee privacy.  
Fusing the biometric signals at an early stage results in a higher accuracy of the protocol at the expenses of a higher complexity.
For this reason, the most used approach, and the one we use in this paper, is 
score level fusion, whereby the match scores from each biometric trait involved in the process are combined to obtain the final result. Score level fusion combines good accuracy and relatively easy implementation.}  

To the best of our knowledge, Gomez-Barrero et al.
\cite{gomez2017multi} have proposed the only previous work on multibiometric privacy protection operating in the encrypted domain.
In their work, the authors present a general framework for multi-biometric 
template protection based on Pailler cryptosystem, in which only encrypted data is handled. 
The authors examine the outcome of the fusion of on-line signature 
and fingerprints, at three different levels of fusion: feature, score and decision levels. 
 The  system presented by Gomez et al. has a low computational cost
 (only one decryption on the server side and no encryptions at verification time), 
moreover they  obtained a good accuracy (EER = 0.12\%), with a required time for a single comparison of about $5\cdot 10^{-4}\textrm{s}$. A drawback with the system described in \cite{gomez2017multi}, is that comparison is carried out on plain data by the server, thus introducing a breach into the security of the system.
On the contrary, in our work we implement also the final comparison step within the \spdz$\,$ framework, to prevent any security loss, even if this choice has a non-negligible cost in terms of complexity (see Section \ref{sec_test_enc}). 
As a further difference, in \cite{gomez2017multi} an Euclidean distance computation (in case of two-modal system) requires $M\cdot F +2$ exponentiations, where $M$ is the number of enrolled samples and $F$ the feature's total number considering all the modalities. 
 In our work, instead, thanks to the \spdz$\,$ system and to the possibility of using integer numbers, we need only $k$ (the length of the feature vector) squares, one of our most expensive operations.

%% file: spdz.tex
Damg\aa rd et al. \cite{damgaard2012multiparty,damgaard2013practical} proposed the MPC framework named \spdz, a two - or multi-party - computation protocol secure against an active adversary corrupting up to $n-1$ of the $n$ players.  
%
This method 
uses multiplicative triples generated offline by using Somewhat Homomorphic Encryption (SHE) to efficiently perform online secret sharing operations.

We assume  the computation is performed 
over a fixed finite field $\mathbb{F}_p$ of characteristic $p$; where $p$ is a prime number. 
Each player $P_i$ 
has an uniform share $\alpha_i \in \mathbb{F}_p$ of a secret key $\alpha$ such 
that
$\alpha=\sum_{i=1}^n\alpha_i ~\textrm{mod}~ p$ (in the following we omit the indication of the modulus operation for simplicity). In this paper we focus on secure two-party computation protocols, then $n=2$ and $\alpha=\alpha_1 + 
\alpha_2$.   
An item $a \in \mathbb{F}_p$ is $\langle \cdot \rangle$-shared 
if the player $P_i$ holds a tuple $\langle a_i, \gamma(a)_i\rangle$ such that $a = a_1+a_2$ and
 $\gamma(a)=\gamma(a)_1+\gamma(a)_2$. In other words, $a_i$ and $ \gamma(a)_i$ are additive 
 secret shares of $a$ and $\gamma(a)$.
 The value $\gamma(a)$  represents the Message Authentication Code (MAC) of $a$.
Any operation involving some variables is also performed on their MAC, so that, at
 the end of the protocol, the MAC 
 is checked before revealing the outcome. 
If one of the parties has a different MAC from the others, the procedure aborts.
 During the description of the protocol, we say that a $\langle \cdot \rangle$ - shared 
 value is \emph{partially opened} if each party reveals to the other one the value $a_i$ 
 but not the associated $\gamma(a)_i$.

 An SPDZ protocol can be divided into two major phases. The preprocessing phase, sometimes referred to as the offline phase, where the 
 system is set up, and the online phase, where the actual computation is 
 performed. 
%
 {In the offline phase,  parties generate a public key and a shared secret key for the SHE scheme.
Then, relying on the homomorphic properties of the SHE, the pre-processing protocol generates $\alpha$ and $\alpha$'s shares,  input shares, shares of tuples for multiplications and squares, and the random share values necessary to evaluate the comparison
  \cite{damgaard2013practical}. Finally each party decrypts his set of pre-processed data by using his secret key share.
  In this paper, we assume that the generation of tuples and inputs has been already  made in the encrypted domain before the protocol starts} and we focus our efforts on the analysis of the online part of the system.

  By using \spdz, linear operations, such 
 as additions and scalar multiplications (see \autoref{tab_linear_operations}), can be performed on the $\langle \cdot \rangle$-shares 
 without interaction; while products between ciphertexts and comparisons need data transmission and proper 
 sub-protocols. 
 \begin{table}[!t]
     \centering
     \renewcommand{\arraystretch}{1.3}
     \caption{Linear operation in \spdz. With $\langle a \rangle$ we indicate 
     the pair $\langle a, \gamma(a)\rangle$, $\langle a \rangle_i$ indicates the 
     pair $\langle a_i, \gamma(a)_i\rangle$. }
     \label{tab_linear_operations}
     \centering
     \begin{tabular}{ccc}
         operation & party 1 & party 2 \\\hline
          $\langle a \rangle +\langle b \rangle$ &  $\langle a \rangle_1 +\langle b \rangle_1$&$\langle a \rangle_2 +\langle b \rangle_2$\\
          $\langle a \rangle -\langle b \rangle$ &  $\langle a \rangle_1 -\langle b \rangle_1$&$\langle a \rangle_2 -\langle b \rangle_2$\\
          $\alpha\cdot\langle a \rangle$ &  $\alpha\cdot\langle a \rangle_1$&  $\alpha\cdot\langle a \rangle_2$\\
           $c + \langle a \rangle$ &  $c+ \langle a \rangle_1$&  $\langle a \rangle_2$\\
          
     \end{tabular}
     \vspace*{-0.4cm}
 \end{table}
 Products and square operations are evaluated through interactive protocols that use multiplication triples generated during the preprocessing phase. Due to lack of space we refer to \cite{damgaard2012multiparty,damgaard2013practical} for implementation details.
%
%
%
%
   Each multiplication requires two transmissions from each party to the other, while each square operation requires only one transmission.

\subsubsection{Comparison}\label{spdz_comp}
Here we show how to compute the outcome of a secure comparison $x < y$, for any two elements
$x,y \in \mathbb{F}_p$, according to the protocol proposed in 
\cite{veugen2015framework}, that has the lowest computational complexity among all the secure comparison protocols proposed so far. 

The comparison computation 
is based on the observation that $\langle x< y\rangle$  
is determined by the truth values of $\langle x< \frac{p}{2}\rangle$, $\langle y< \frac{p}{2}\rangle$,  and $\langle(x-y) \mod p <\frac{p}{2}\rangle$,
where $\langle x < y \rangle$ indicates the
share values of the outcome of $x<y$.
By choosing $p$ so large that both inputs are lower than 
$\frac{p}{2}$,
it is sufficient to evaluate only $\langle(x-y) \mod p <\frac{p}{2}\rangle$.

Given $z = x-y$,   then $\langle x < y \rangle$ can be easily 
computed as $1- \langle z < \frac{p}{2} \rangle$.
We observe  that if $z > \frac{p}{2}$ then $2z > p$. Since we 
work on $ \mathbb{F}_p$, we have that 
$ 2z \mod p = 2z-p$ and it is 
odd; else if $z < \frac{p}{2}$ than
$2z <p$ and it will be even because we do not need any modular operation. 
Therefore, to establish if $z$ is larger or  smaller than $\frac{p}{2}$ 
we need to determine only the last significant bit of $2z$.  
To compute the last significant bit of $2z$, we use a value $\langle r
\rangle$ shared by parties both as integer and as a bit array. The value $r$ 
along with  its bit decomposition are pre-computed off line.
We indicate as 
$r_0 r_1\ldots r_\ell$ the bits of $r$ and with $\langle r_i \rangle $ their
shared values.

First of all we compute $\langle s \rangle = \langle 2z + r \rangle$, then $s$ 
is partially opened. If $s< p $ then the last significant bit of $2z $ is 
equal to  $s_0 \oplus  r_0$, otherwise it is equal to $ 1- (s_0 \oplus r_0)$.

Since we work in the field $\mathbb{F}_p$, $s<p$ iff $s<r$. 
By recalling that $s$ is known to 
both parties, we can easily obtain  a $\langle \cdot \rangle $-share 
of $\delta$, the truth value of $ \langle s < r \rangle$ (i.e. 
$\langle \delta \rangle = \langle s < r \rangle$) working on the bits of $s$ and on the shared bits of $r$.
Then we use the following procedure to calculate  $\langle \delta \rangle = \langle s < r 
\rangle$.
\begin{itemize}[nosep,noitemsep]
	\item[]If $s_0=0$ then  $\langle \delta \rangle = \langle r_0\rangle$ else  $\langle\delta\rangle=\langle 0 
\rangle$.
	\item[]For all $i < \ell-1$ 
	\item[]\qquad if $s_i = 0$ then  $\langle \delta \rangle = \langle  r_i \rangle+ \langle \delta \rangle \cdot \langle 1-r_i \rangle$ 
	\item []\qquad else  $\langle \delta \rangle = \langle  r_i \rangle \cdot \langle \delta 
	\rangle$.
\end{itemize}

Now  $\langle z<\frac{p}{2} \rangle $ can be easily calculated as
\begin{equation}\label{eq_comp}
\langle \delta \oplus s_0 \oplus r_0\rangle= \langle \delta \rangle -\langle s_0 \oplus r_0\rangle -\langle \delta \rangle\cdot \langle s_0 \oplus r_0 \rangle
\end{equation}
Since $s_0$ is known in our implementation,
{\small
\begin{equation}
	\langle \delta \oplus s_0 \oplus r_0\rangle=\left\{
	\begin{array}{ll}
	\langle \delta\rangle +\langle r_0\rangle -2\cdot\langle \delta\rangle\cdot\langle r_0\rangle &\textnormal{if }s_0=0,\\
	1+2\langle\delta\rangle\cdot\langle r_0\rangle-\langle r_0\rangle -\langle\delta\rangle & \textnormal{if }s_0=1.
	\end{array}
	\right.
\end{equation}}
%
%

\emph{Complexity}.
This protocol requires 
one multiplication for each iteration  plus one for the last step in (\ref{eq_comp}). Therefore the complexity depends on the bit length 
of $r$ and it is equal to 
$\ell$ multiplications, which require $2\ell$ transmissions.

%% file: Biometrics.tex
In our work we compare two different authentication systems: an iris authentication protocol and a multimodal system relying on the fusion at the score level of iriscode and eigenfaces. 
Both protocols are developed  in the encrypted domain by relying on the \spdz~ framework.
Different and more accurate protocols exist, but their privacy preserving implementation have such a high complexity to make them impractical.
 In the following subsections, we first describe the standalone iris (Section \ref{sec_iriscode_descrizione}) and face
 (Section \ref{sec_face_descrizione}) protocols in the plain domain, then the 
 main characteristics of a
 general multimodal recognition protocol (Section \ref{sec_multibio}).
\subsubsection{Iris recognition}
\label{sec_iriscode_descrizione}
\input{iris}

\subsubsection{Face recognition}
\label{sec_face_descrizione}
\input{face}

\subsubsection{Multi-biometric score level fusion}
\label{sec_multibio}
\input{fusione}

%% file: iris.tex
In our implementation, we use the iriscode template
proposed for the first time by Daugman in \cite{daugman2004} and then modified 
by Masek in  \cite{masek2003recognition}. 
The description of the entire extraction process is out of the scope of this paper, 
therefore we limit the presentation to the details which are relevant for the current work.

An iriscode is a bit vector of length $N$ depending on the radial $r$ and angular $\theta$
resolutions used during template extraction. 
 The extraction process  outputs also  a bitwise noise mask.  
 The noise mask represents the regions of the iris altered by noise, e.g. by eyelashes end 
 eyelid.
%
 To compare two templates (the query and the probe enrolled in the database) the authentication process 
 relies on a weighted
Hamming distance, where the weights depend on the noise mask bits. In this 
way only significant bits are used to calculate the distance between the two 
templates.
Given  {the template length} $N= 2\cdot r \cdot \theta$,  we indicate by $F_1=f_{1,1}\ldots f_{1,N}$ and $F_2=f_{2,1}\ldots f_{2,N}$ the iris templates and 
$M_1=m_{1,1}\ldots m_{1,N}$ and $M_2=m_{2,1}\ldots m_{2,N}$ the corresponding noise 
masks. {We assume that the value $m_i=1$ in the mask vector indicates that the bit is affected by noise  and must be excluded from the computation. 
Moreover,  we indicate with $\overline{a} = 1-a$ the  negation of a feature bit $a$.  }
The weighted Hamming Distance (HD) can be calculated as:
\begin{eqnarray}\label{eq_hamming_dist}
  HD &=& \di{\frac{\| (F_1 \oplus F_2 )\wedge ( \overline{M}_1\wedge \overline{M}_2)\| }
  	{N-\|M_{1} \vee M_{2}\|}}\nonumber\\ 
  &=& \di{\frac{\sum_{j}^N\left[(f_{1,j} \oplus f_{2,j} )\wedge ( \overline{m}_{1,i}\wedge \overline{m}_{2,i})\right ]}{N-\sum_{j=1}^Nm_{1,j} \vee m_{2,j}}} .
\end{eqnarray}


%% file: face.tex

{Face templates can be generated by relying on  the \emph{eigenfaces} method proposed by Tuck and Pentland in 
\cite{turk1991face} providing a set of 
facial characteristics that can be used to describe all the faces into the database (the \emph{face-space}),
as eigenvectors do in linear algebra. 
The template associated to a face $\Gamma$, is the projection of $\Gamma$ on the face space 
 $\Omega =[\omega_1\ldots\omega_k]$. Each $\omega_j$ 
describes the contribution of the corresponding eigenface, $e_j$, in representing the input 
image.
}
In order to find the image that best matches with $\Gamma$,  the algorithm looks for the projection vector $\Omega_j$ among all database images, 
 that minimizes the 
Euclidean distance 
\begin{equation}\label{eq_Euclidea}
ED=\| \Omega - \Omega_j\| =\sqrt{\sum_{i=0}^k \left(\Omega_i-\Omega_{j,i}\right)^2 }.
\end{equation}

%% file: fusione.tex
In this work we choose to use multimodal fusion (face and iris) at score 
  level. {This is motivated by the fact that, following the considerations in \autoref{sec_multiChiaro}, face and iris recognition protocols allow to easily compute the match score
  in the encrypted domain. Moreover, we underline  that both iris and face can be acquired at the same time in real applications, for example by using a smartphone camera.} 
   We summarize our  fusion system in \autoref{fig_fusione}.
    \begin{figure}[ht]
      \centering
      \includegraphics[width=\columnwidth]{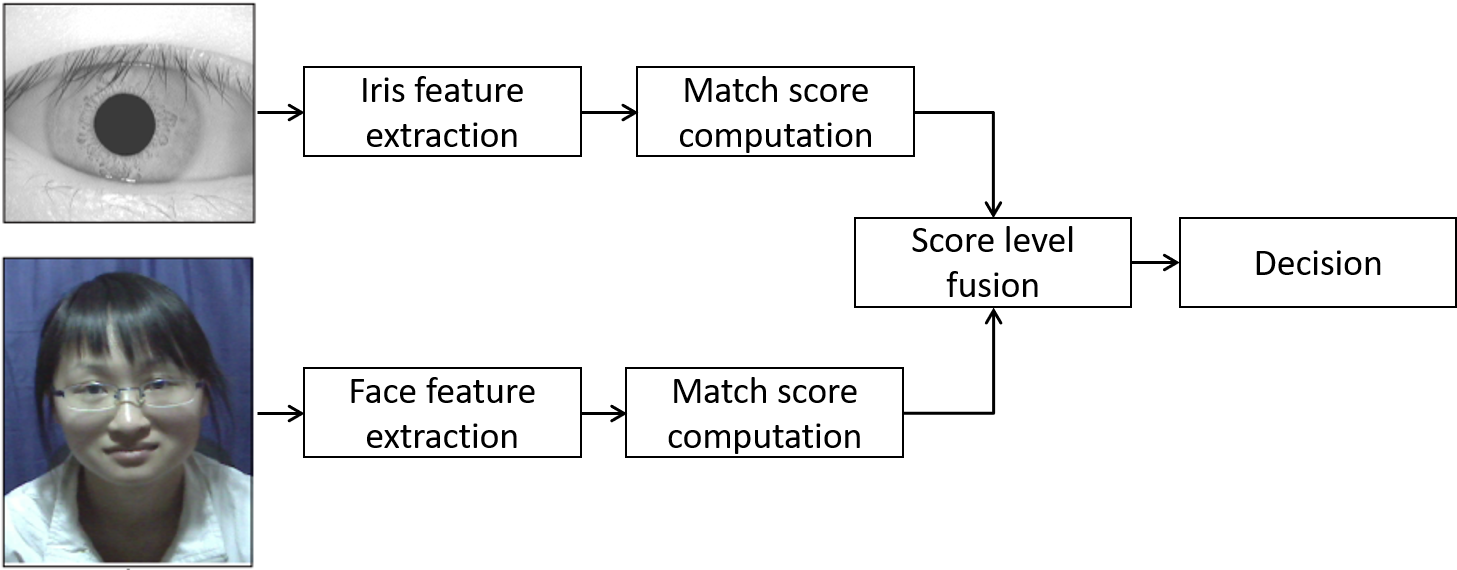}
      \caption{General scheme of our fusion protocol.}
      \label{fig_fusione}
    \end{figure}
    
  Match scores generated by iris and face classifiers are characterised by a different range of values: the output of the iris protocol 
  is a real number in $[0,1]$, while the output of the face recognition protocol is a squared number in $[0, M], M \in \mathbb{R}$.
  Many ways have been proposed to overcome the problems generated by the 
  differences between the scores provided by different biometric recognition systems (see \cite{ross2006handbook} for more details). Among 
  them, due to the characteristics of our MPC system, we choose a linear 
  combination of the scores. Furthermore,
  to normalize the face matching score, we choose to use a min - max normalization method 
  \cite{ross2006handbook}. 
  \begin{equation}
    \mathrm{face\_score\_norm}=\di{\frac{\mathrm{score}-\mathrm{min}_\textit{face}}{\mathrm{max_\textit{face} - min_\textit{face}}}},
  \end{equation}  
  
\noindent where $\mathrm{min}_\textit{face}, \mathrm{max}_\textit{face}$, indicate the minimum and maximum values of 
 the face range. Since $min_\textit{face} =0$, 
  if $i\in [0,1]$ is the Hamming distance resulting from 
  iris match, $f$ the Euclidean distance between faces, and $R$ the maximum
  value of the face recognition system, multimodal recognition corresponds to checking if the following inequality holds: 
  
  \begin{equation}\label{eq_fusionRule}
    \alpha \cdot i +\beta \cdot \di{\frac{f}{R}} < t
  \end{equation}
  
  where $\alpha, \beta \in [0,1]$ are proper weights and $t\in [0,1]$ is the threshold.

The choice of the parameters $\alpha, \beta, T$ determines the trade-off between equal error rate (EER) and computational complexity
 (Section \ref{sec_test_clear}).
As in \cite{connaughton2013fusion,gomez2017multi} we choose $\beta=1-\alpha$ for the tests in the plain domain.
 More details are provided 
     in Section \ref{sec_implements}.

%% file: implem.tex

   \subsection{Iris authentication protocol in encrypted domain}
   \label{sec_iris_Encrypted}
   The \spdz$ \,$ protocol supports operations modulo $p$. Each binary element of an iris feature is encrypted as a modulo $p$ integer
   $ \langle a \rangle$-share.  For this reason, to  implement the Hamming distance computation as in Equation (\ref{eq_hamming_dist}),
   we must implement logical operations $\oplus, \vee, \wedge$ as a combination of integer operations $+, - , \cdot$.
   {The correspondence between binary and integer operations\ is detailed in \autoref{tab_correspondence}}.
   
     \begin{table}[h!]
	\centering
	\renewcommand{\arraystretch}{1.3}
	\caption{Correspondence table between binary and integer operations}
	\label{tab_correspondence}
	\begin{tabular}{cc}
		Binary & Integer\\\hline
		$a\oplus b$&$a + b-2\cdot a\cdot b$\\
		$a \wedge b$&$a \cdot b$\\
		$a \vee b $&$ a +b -a\cdot 
		b$\\
		$\overline{a}$&$1-a$
	\end{tabular}
	
\end{table}

    Let  $F_1=f_{1,1},f_{1,2},\cdots, f_{1,N}$ and $F_2=f_{2,1},f_{2,2},\cdots, f_{2,N}$ be two binary iris feature templates, where $N$ is the number of features, we 
    indicate with $ \langle F_i \rangle$ the  vector containing  shares of each 
    element, i.e. the vector $ \langle f_{i,1 }\rangle, \langle f_{i,2} \rangle, \cdots \langle f_{i,n 
    }
    \rangle$ for $i=1,2$.

       Since $\overline{M}_1\wedge \overline{M}_2$ is equivalent to $\overline{M_1 \vee M_2}$, the Hamming distance in Equation (\ref{eq_hamming_dist}) can be
    computed 
    as:
{\small\begin{equation}\label{hamming_dist_cifrata}
\di{\frac{\sum_{i=1}^N\left \{ (f_{1,i} + f_{2,i} -2\cdot f_{1,i}\cdot f_{2,i})\cdot \left [1- (m_{1,i}\vee m_{2,i})\right ] \right \}}{N-\sum_{j=1}^N (m_{1,i}\vee m_{2,i})}}, 
\end{equation}}
where $m_{1,i}\vee m_{2,i}=m_{1,i}+m_{2,i}-m_{1,i}\cdot m_{2,i}$.

In \spdz, as well as in other MPC protocols, division is a very expensive operation. Hence, instead of evaluating the division and then compare the distance with an acceptance threshold, the denominator is multiplied by the threshold before the comparison. 
By letting 
{\small\begin{equation}\label{eq_numIris}
num = \sum_{i=1}^N\left\{ (f_{1,i} + f_{2,i} -2 f_{1,i}f_{2,i}) \left[1- (m_{1,i}\vee m_{2,i})\right] \right\}
\end{equation}}
and
{\small\begin{equation}\label{eq_demIris}
den = N-\sum_{j=1}^N (m_{1,i}\vee m_{2,i}),
\end{equation}}
the authentication check corresponds to 
 \begin{equation}
\mathit{num} < t \cdot \mathit{den}.
\end{equation}

As it can be seen from Equation (\ref{eq_numIris}), many share multiplications are needed to calculate the numerator, namely, $N$ multiplications for each $F_1\oplus F_2$,  $M_1 \wedge M_2$ and $(F_1 \oplus F_2 )\wedge ( M_1\wedge 
M_2)$. To compute the denominator, we can reuse the $m_{1,i}\vee m_{2,i}$ already computed for the numerator.
 Multiplication between shares requires data transmission, slowing down the computation.
 To optimize the protocol, we split the multiplication protocol into two parts.
First we calculate
$\langle \varepsilon_i\rangle =\langle f_{1,i}\rangle - \langle a\rangle$  
and $\langle \delta_i\rangle =\langle f_{2,i} \rangle - \langle b\rangle$, for all $i=1\ldots 
N$. Then, to partially open the values, both server and client exchange  shares  
by using a  packet for all $\varepsilon$'s and one for all $\delta$'s. In this 
way we need only two transmissions for $N$ multiplications.

\emph{Complexity}.
Computing $F_1 \oplus F_2$ and $M_1\wedge M_2$ requires $N$ multiplications each, 
one for each element of the template; moreover, $N$ multiplications are required to compute $(F_1 \oplus F_2 )\wedge ( M_1\wedge M_2)$.
The total cost associated to the computation of $\mathit{num}$ is $3N$ multiplications but, as we explained above, we need only $6$ transmissions.
Computing $\mathit{den}$ (Equation (\ref{eq_demIris})) has a negligible complexity since $m_{1,i}+m_{2,i}-m_{1,i}\cdot m_{2,i}$ 
has already been calculated for all $i$ in Equation (\ref{eq_numIris}). Moreover, we need a multiplication between $\mathit{den}$ and $t$, and $\ell$ multiplications for the 
comparison, where $\ell$ is the number of bits necessary to represent a modulo $p$ integer (see Section \ref{spdz_comp}). Consequently, we need $3N+\ell+1$ 
multiplications but only $2\ell+7$ transmissions for the iris protocol.

\subsection{Face authentication in encrypted domain}
\label{sec_faceEncry}
{As for the iris, we assume that  face features have already been computed according to the protocol described in Section \ref{sec_face_descrizione}, obtaining a set of $k$ real features $\Omega_i$ that have been rounded to represent them in $\mathbb{F}_p$ (in the next we avoid the round operator for simplicity).
Given the projection $\Omega$ of the query face image, the face-based biometric authentication protocol must evaluate the Euclidean distance  $ED$ 
as in Equation (\ref{eq_Euclidea}), 
{and check if it is lower than a threshold $t$}.}
Considering that the square root cannot be evaluated efficiently in \spdz, we instead compare 
the Squared Euclidean distance (SED) against the squared threshold:
\begin{equation}\label{eq_sed}
\sum_{i=1}^k(\Omega_i-\Omega_{j,i})^2 < t^2.
\end{equation}


{In equation (\ref{eq_sed}), $\Omega_i$ indicates an element of the face $\Omega$ and $\Omega_{j,i}$ the $i$-th element of the projection $\Omega_j$.}
Moreover, as we did for the Hamming distance, we separate the square computation into two parts, so we need only one transmission to calculate SED. 

\emph{Complexity.}
The computation of the Squared Euclidean distance requires the evaluation {of} $k$ squares that can be parallelised, hence only one transmission is necessary. 
 For the comparison 
we need $\ell$ products, as in the iris authentication protocol. Considering that squares and products have similar complexity, the complexity of the protocol is given by $k+\ell$ products.

\subsection{Fusion in encrypted domain}
\label{sec_fusionEncry}
We now describe our solution to implement the fusion protocol in the encrypted 
domain. As outlined in Equation (\ref{eq_fusionRule}), we use a linear 
combination of the matching scores;
to avoid performing divisions, we evaluate
  \begin{equation}\label{eq_fusionRule_enc}
    \alpha \cdot {num} \cdot R+\beta 
    \cdot \textrm{SED}\cdot {den}< T \cdot den \cdot R,
      \end{equation}
 where $num$ and $den$ stand for the numerator and denominator of the iris Hamming distance, $i$ in (\ref{eq_fusionRule}), while $\textrm{SED}$, $R$
  {and $T$} stand for squared Euclidean distance score, face maximum 
range, and threshold.   

   
         The \spdz$\,$ framework does not allow the use of non-integer numbers, so $\alpha$, $\beta$ and $T$
      are scaled and approximated to integers in the interval $[0,10]$. 
      {We chose this interval because it is accurate enough to obtain the same results achieved in the plain domain, and the resulting bitlength is small enough to make it possibile to represent $\alpha \cdot {num} \cdot R+\beta 
      	\cdot \textrm{SED}\cdot {den}$ and $T \cdot {den}\cdot R$ in $\mathbb{Z}_p$.}

      \emph{Complexity} 
      The previous formula requires three multiplications and six transmissions 
      that cannot be run in parallel. Moreover, it needs $\ell$ multiplications for the
      comparison. In total, the linear fusion requires $\ell+6$ multiplications and $2\ell+12$ 
      transmissions, plus the multiplications necessary to compute the Hamming and the squared Euclidean distances.
      The total complexity of the full multimodal protocol is $3N+\ell+6$ 
      multiplications and $k$ squares, while it requires only $2\ell+19$ transmissions.
    
    \begin{table}[ht]
      \renewcommand{\arraystretch}{1}
        \centering
        \caption{Complexity Summary. We underline that transmission number depends only on $p$'s bitlength $\ell$. }\label{tab_complexity}
        \begin{tabular}{cccc}
            & Multiplication & Squares & Transmissions\\\hline
         Iris  &$3N+\ell+1$&0& $2\ell+7$\\
         Face & $\ell$&$k$ &$2\ell+1$\\
        Multimodal&$3N+\ell+6$&$k$&$2\ell+19$
        \end{tabular}
    \end{table}
    
\subsection{Protocol security}\label{sec_securityDemonstration}

Relying on \spdz$\,$ tool, our protocols is secure in the UC model if at least one of the two parties is honest.
In fact, according to \cite{damgaard2013practical}, our \spdz -based protocol is secure against $n-1$
malicious adversaries, where $n=2$ in our two-party computation scenario. 
The offline phase does not depend on the functionality evaluated and its security demonstration against active adversaries in the UC model is provided in \cite{damgaard2013practical}. 
The security demonstration  of the online protocol is provided in the following theorem.

\begin{theorem}
	The online \spdz ~implementation of \acronym~is computationally secure against any static adversary corrupting at most $1$ party if $p$ is exponential in the security parameter.
\end{theorem}
\begin{proof}
	\newcommand{\simul}{$\mathcal{S}_{\scriptsize\texttt{ONLINE}}$}
	\newcommand{\fonline}{$\mathcal{F}_{\scriptsize\texttt{ONLINE}}$}
The proof follows the security demonstration of the online \spdz~ protocol in \cite{damgaard2013practical}.
We rely on the simulator \simul$\,$ defined in \cite{damgaard2013practical} to work on top of the ideal multibiometric authentication functionality \fonline, such that the adversary cannot distinguish among the simulator using the real function \fonline$\,$ and the real \spdz-based implementation using multiplication triples generated offline. 
Input values broadcasted by both the simulator and honest players are uniform and it is not possible to distinguish among them.
During execution, interaction with player is performed only during multiplication and squaring where partial opening reveals uniform values for both honest parties and simulator. 
Also MACs have similar distribution in both the protocol and the simulation.
If the protocol does not abort due to a cheat detection, both the real and the simulated runs output the decision bit. 
In the simulation the decision bit is obtained by a correct evaluation of the multi-biometric function on the inputs provided by the player. In real \spdz-based implementation the adversary can cheat in the MAC check with probability $2/p$. 
Hence the probability that the adversary can distinguished the simulated environment from the real one is negligible if $p$ is exponential.  
The adversary is not able to obtain the inputs of the
honest player because if the protocol does not abort, he can
observe only its input, the input shares received by the other
party and the final result. To obtain the original inputs of
the honest party, the adversary should be able to solve the inequality in
(\ref{eq_fusionRule}), which has $N_i + N_f$ unknown variables for the
server and $N_i+N_f +3$ for the client, where $N_i$ and $N_f$ are the
number of features used to represent iris and face respectively.
\end{proof}

%% file: tests.tex
\section{System tuning}\label{sec_tuning}
{In  this section, we present the results of the tests performed on plain data. We will use such results to choose the best parameters to build an efficient protocol working in the encrypted domain. }
Tests have been carried out on the  ``CASIA-IrisV1'' database for irises  and ``CASIA-FaceV5 part 1'' database for faces,
 both collected by the Chinese Academy of Sciences'
 Institute of Automation (CASIA) 
 \cite{casiaIris,casiaFace}.  
 
 The CASIA-IrisV1 database for irises \cite{casiaIris} contains $756$ grey-scale eye images 
 with $108$ unique irises (or classes) and $7$ images for each of them. 
As in \cite{masek2003recognition}, we used a subset of the database for the tests, retaining only those images 
{in wich the algorithm has well separated iris region from sclera and pupil.}
The resulting database contains  
 $625$ eye images. 
 
 The CASIA-FaceV5 Databases for faces part 1 \cite{casiaFace} contains $500$ 
 face images of $100$ subjects. The face images are captured using Logitech USB camera in one session.
  All face images are $16$ bit color BMP files and the image resolution is $640 \times 480$ pixels.
 
  We implemented and tested  our \spdz-based iris and multibiometric protocols on a desktop equipped with 8GB RAM 
 processor Intel Core i3 CPU 550 @ 3.20 GHz Quad-Core
  running Ubuntu 14.04 LTS (64 bit) operative system.  We developed the test using C++ language with GMP 
   free library for arbitrary precision arithmetic, operating on signed integers, rational numbers, 
   and floating-point numbers.
   
   To implement the \spdz $\,$ protocol we chose the $46$ bit  prime number $p= 67280421310721$, which is big enough 
   to allow all the needed modular operations and comparisons, and guarantee the security of the protocol. Server and client run on the same computer, and we used a socket
   to simulate the transmission channel.

 \subsection{Parameter optimisation}
 \label{sec_test_clear}
 
We performed tests on plain data, running the authentication protocol on each single biometric and then by fusing eigenfaces 
and iriscodes.

\paragraph{Iris} For testing the iris authentication protocol, we have chosen a 
 radial resolution $r$ ranging from $4$ to $20$ and an angular resolution $\theta$ 
 between 
 $100$ and $200$ 
 (\autoref{Tab_noshift}). As said above,
 we tested the protocol on $625$ eye images and each one has been compared with all the 
 others.
 To perform the tests in the plain domain, we used the Matlab code provided by L. Masek (iris recognition source code 
 \cite{Masekcode}) as part of his work \cite{masek2003recognition}. 
 As we can see from  \autoref{Tab_noshift}, the best accuracy is achieved by letting the angular resolution be equal to $160$ angles and radial resolution equal to 20 corresponding to an iris feature vector of length 6400. 

 To reduce the EER, Masek (\cite{masek2003recognition}) shifts $n$ times 
the iris templates keeping the lowest Hamming distance score.
 In \spdz$\,$ this operation is 
computationally very expensive, so we could not afford it. 

\begin{table}[ht]
	\centering
	\renewcommand{\arraystretch}{1.3}
	\caption{Iriscode EER (\%) without shifting, as a function of different values of $r$ and $\theta$.}
	\label{Tab_noshift}
	\begin{tabular}{ c @{\qquad\quad}cccccc}
		\multirow{ 2}{*}{$r$}    & \multicolumn{6}{c}{Angular Resolution $\theta$}\\
		
		&	100&	120&	140	&160&180&200	\\\hline
		
		4	&8.19&	6.43&	4.37	&3.34&	3.31	&3.10\\
		6&	6.88&	5.05&	3.01&	2.45&	2.71&	3.01\\
		8&	6.13&	4.42&	2.69&	2.19&	2.58&	4.36\\
		10&	6.31&	4.03&	2.64&	2.44&	2.54&	3.92	\\
		12&	5.96&	4.10&	2.56&	2.14&	2.59&	3.71\\
		14&	5.71&	3.85&	2.54&	2.17&	2.58&	3.27\\
		16&	5.49&	3.79&	2.32&	2.13&	2.51&	3.31\\
		18&	5.71&	3.61&	2.46&	2.31&	2.41&	3.18	\\
		20&	5.77&	3.74&	2.20&	2.08&	2.41&	3.13
	\end{tabular}

\end{table}

 \paragraph{Face} We have implemented the eigenface protocol by using the Open Source Computer Video (openCV) library\footnote{\url{http://opencv.org/about.html}} 
 and Matlab. Face images are $640 \times 480$ 
 pixel and
  we transformed them in  $256$ grey level images. 
  The protocol has been tested on $500$ images.
We used the algorithm provided in the openCV library 
to build 
 $k$ eigenfaces with $k=1 \ldots 10$. Each image is thus represented by a projection vector of length $k$. Each projection element is a $16$-bit integer and
the squared Euclidean distance has been calculated by using Matlab. We observed that the use of more than 5 projections does not provide any significant improvement (see \autoref{Tab_Face}). 

\begin{table}[hbt]
	\renewcommand{\arraystretch}{1.3}
	\centering
	\caption{
		Eigenface EER (\%) values, as a function of the number of projections.}
	\label{Tab_Face}
	\begin{tabular}{cccc}
		$k$ & EER (\%)&$k$ & EER (\%)\\\hline
		$1$&$28.77$&$6$&$17.01$\\
		$2$&$17.37$&$7$ &$16.19$\\
		$3$ &$16.62$& $8$&$16.51$\\
		$4$&$16.59$&$9$&$16.38$\\
		$5$ &$16.08$&$10$&$16.09$
	\end{tabular}
\end{table}

\begin{table*}[!htb]
	
	\centering
	\renewcommand{\arraystretch}{1.3}
	\caption{EER of the multimodal biometric authentication protocol. The first three columns show iris's parameters: feature's number (N),
		radial resolution ($r$), and angular resolution $\theta$. Fourth column represents Iris authentication system's EER ($\%$). 
		All the others columns contain  the EER's obtained by fusing an iris template of length $N$ and a face template with  $k\in\{1\ldots 10\}$ eigenfaces. We highlighted in bold the configurations that we have selected for the tests under encryption.
	}
	
	\label{Tab_Fusion_totale}
	\begin{tabular}{ccc|c|cccccccccc}
		\multicolumn{3}{c|}{Iris Parameters}& &\multicolumn{10}{c}{Number of Eigenfaces ($k$)} \\
		N&r&$\theta$&iris&1&2&3&4&5&6&7&8&9&10\\\hline
		6400&20&160&2.08&1.17&\textbf{1.15}&1.02&1.25&1.25&1.24&1.31&1.37&1.4&1.41\\
		5760&16&180&2.51&1.26&\textbf{0.98}&1.01&1.22&1.38&1.36&1.43&1.47&1.49&1.50\\
		5600&20&140&2.20&1.20&1.08&1.19&1.18&1.34&1.28&1.36&1.38&1.39&1.40\\
		4800&20&120&3.74&1.90&1.97&1.65&1.98&2.04&2.15&2.11&2.17&2.2&2.21\\
		3840&12&160&2.14&1.84&1.52&1.50&1.45&1.63&1.74&1.76&1.77&1.78&1.78\\
		3600&10&180&2.54&1.36&1.23&\textbf{0.97}&1.65&1.82&1.99&2.07&2.15&2.17&2.19\\
		3360&12&140&2.56&1.51&1.32&1.38&1.31&1.56&1.61&1.58&1.63&1.67&1.69\\
		2560&8&160&2.19&1.50&1.24&1.19&1.15&1.39&1.49&1.59&1.6&1.61&1.62\\
		2400&6&200&3.01&2.01&1.83&1.89&2.13&2.37&2.56&2.66&2.7&2.72&2.73\\
		2160&6&180&2.71&1.92&1.74&1.82&1.98&2.04&2.09&2.07&2.13&2.16&2.18\\
		1920&6&160&2.45&1.47&1.42&1.43&1.57&1.66&1.95&1.92&1.95&1.96&1.97\\
		1600&4&200&3.10& \textbf{2.01}&\textbf{1.87}&1.85&2.41&2.37&2.56&2.67&2.69&2.71&2.71\\
		1280&4&160&3.34&2.29&1.89&2.22&2.26&2.51&2.80&2.80&2.92&2.98&3.01
	\end{tabular}

\end{table*}

\paragraph{Multimodal}
We have evaluated the efficiency of the fusion protocol in the plain domain, by fusing the outcomes of face and iris sub-algorithms through a  Matlab implementation. 
From \autoref{Tab_noshift}, we have chosen some relevant iris configurations, based on the achieved EER or number of features. First of all, to better compare with  the best iris result,
we chose $r=20$ and $\theta=160$ resulting in $N=6400$, then for each $\theta$,  we looked for the best accuracy under $4\%$,
and finally we chose those configurations with EER similar to the previous ones but less features.
Moreover, we varied $\alpha$ in the interval $[0,1]$ and the number of eigenfaces $k$ from $1$ to $10$.
  \autoref{Tab_Fusion_totale} summarizes our results, showing the EER for each $N$ and $k$.

As shown in \autoref{Tab_Fusion_totale},
the same accuracy of the $6400$ stand alone iris protocol (2.08\%) can be reached with many different multi-biometric configurations, e.g. by using $N=3600$ and $k=7$ or even by using only $1600$ iris features and $k=1$. 
For the tests in the encrypted domain, between the two configurations with the same accuracy, we chose  the last one, since it has lower bandwidth and computational complexity 
(see \autoref{tab_complexity} and \autoref{tab_trans}). 
We can also notice that keeping $N=1600$, but using 2 features for face 
representation, we can lower both accuracy and complexity. 
Generally, by using two eigenfaces, the best possibile accuracy is provided with  $5760$ iris 
features, however the same performance are obtained also by the combination of $3600$ iris features and 
three face features.
For this reasons, we tested several 
configurations in the encrypted domain, as summarised in \autoref{Tab_Fusion}.

 \section{Complexity of the \spdz~ protocol}
 \label{sec_test_enc}
 
 We evaluated the computational complexity of our implementation of the \spdz$\,$ protocol (Section \ref{sec_implements}), by using the parameters chosen in the previous section (see \autoref{Tab_noshift} and 
 \autoref{Tab_Fusion}).
Execution times are heavily affected by the number of multiplications. As we said 
in Section \ref{sec_iris_Encrypted}, when possible, we performed a single 
transmission, by packing data. 
To calculate the execution time, reported in \autoref{tab_iris_spdz}, we used the \emph{clock}
 function, measuring the CPU time of the process. 
\begin{table}[ht]
	\centering  
	\renewcommand{\arraystretch}{1.3}
	\caption{Equal Error Rate of iris and multimodal biometric  authentication protocols for different settings;  $\alpha$, $t$ respectively stand for fusion coefficient and threshold.}
	
	\label{Tab_Fusion}
	\begin{tabular}{cccccc}
		{Iris}    &	 \multicolumn{2}{c}{EER}&	Face& \multicolumn{2}{c}{Fusion parameters}\\
		N  &   Iris  (\%) & Fusion (\%)&{$k$}		&{$\alpha$}	&{$t$}\\\hline    
		$1600$ &$3.10$&	$2.01$&1&	$0.80$&	$0.35$\\

		$1600$&$3.10$&	$1.87$&2&	$0.55$&	$0.25$\\
		$3600$&$2.54$&$0.97$&3&$0.55$&$0.25$\\
		$5760$&$2.51$&	$0.98$& 2&	$0.80$&	$0.35$\\
		$6400$&	$2.08$&$1.15$&2&	$0.80$&	$0.35$
	\end{tabular}
\end{table}
%
%

%
   
   {In \acronym~the number of transmission rounds depends only on the bitlength $\ell$ of the prime number $p$  and not on the feature configuration, as it can be seen from \autoref{tab_complexity}.
   On the contrary, the amount of data transmitted by each party also depends on the number of features used in the protocol. 
   In fact, the iris authentication protocol has a bandwidth of  $(6N+2\ell+2)\cdot\ell$ bits, 
    while the multimodal protocol bandwidth is  $\ell\cdot (6N+k+2\ell+12)$ bits. 
   Since the complexity of the iris protocol is much higher than that of the face-based authentication protocol, the overhead introduced by the multimodal biometric authentication is of few bytes, as it can be seen from \autoref{tab_trans}. 
  For this reason, the communication complexity remains almost constant switching 
  from the iris to the multimodal protocol.
   } 
\begin{table}
     \renewcommand{\arraystretch}{1.3}
 \caption{Communication complexity for the iris and multimodal protocols. 
 It is important to notice that adding few eigenfaces incerases the bandwidth by a few bytes only.}\label{tab_trans}
  \centering
	\begin{tabular}{ccccc}
{Iris}    &	 \multicolumn{3}{c}{bandwidth (KB)} & Face \\
      $N$ & iris & multimodal & overhead &$k$\\\hline
   
  1600& 53.24   &53.30 &0.06&1\\
  3600&119.16 &119.23 &0.07&3\\
  5760&190.35 & 190.42 &0.07&2\\
  6400&211.44  &211.51 &0.07&2
       
     \end{tabular}
   \end{table}
   
The main goal of our work was exploiting multimodality to reduce complexity while maintaining the same accuracy of the iris-based protocol.
  Moreover, our analysis shows that the multimodal protocol can also be used to lower the EER without a significant loss in terms of complexity. In the following, we discuss both cases. 

\paragraph{Improved efficiency}
 {The running time of the stand alone iris authentication protocol ranges from $0.03$s 
for $1600$ bits, up to $0.12$s for a $6400$ bit-long template
in the malicious setting  (see \autoref{tab_iris_spdz}), while Luo et al. protocol  \cite{luo2012efficient} with masks needs $2.5$s for $9600$ bits and $0.56$s for $2048$ bits in the semi-honest setting.
Moreover from \autoref{Tab_Fusion} and \autoref{tab_iris_spdz}, it is evident that \acronym~can provide the same accuracy of the best stand alone iris 
protocol, but with lower execution time and computational complexity. 
As a matter of fact, the 
best EER for the standalone iris protocol is $2.08\%$ for $6400$ features corresponding to $19246$ multiplications (\autoref{tab_complexity}) in $0.12\textrm{s}$, 
while in the fusion configuration for $1600$ iris features and $1$ eigenface feature, we need 
only $8744$ multiplications (see \autoref{tab_complexity}, where we consider squares as multiplications) to obtain an EER equal to $2.01\%$ in about $0.03$ seconds. 
On the contrary, the number of required transmissions increases from $2\ell+7$ to $2\ell+19$ (\autoref{tab_complexity}), 
but it depends only on the bit length of $p$.
}

\paragraph{Improved accuracy}
{As an alternative to improve the computational complexity, the use of two biometries instead of one can be exploited  to achieve a higher accuracy, at the cost of a slight increase of complexity with respect to the iriscode protocol. 
In fact, as shown in \autoref{tab_complexity}, complexity depends heavily on the number of iris features, however by adding two eigenfaces it is possible to decrease the EER rate, while the 
number of multiplications increases only from $3N+\ell+1$ to $3N+\ell+6+k = 3N+\ell+8$ (as usual we consider squaring to be equivalent to multiplication).
More generally, when we move to multimodal authentication, the total CPU time slightly increases with respect to the unimodal iris protocol, but the EER always decreases;
by adding one more eigenface ($k=2$) to the $1600$ iris feature configuration considered above, we can have a better EER ($1.87\%$) with the same time complexity (0.03 seconds). 
For the case of $5760$ bit long iris template, the EER passes from  from $2.1\%$ for the unimodal authentication to $0.98\% $ for the bimodal case with $k=2$ (\autoref{Tab_Fusion}). 
Finally, keeping $0.98\%$ as target accuracy, we highlight that we can reduce $N$ to $3600$ at the cost of an additional feature in the face representation ($k=3$). In this case, computational complexity goes from $36926$ to $19596$ multiplications  and time complexity decreases from $0.109$s to  $0.05$s (\autoref{tab_trans} and \autoref{tab_iris_spdz}). }

%
 
%
%
%
   
\begin{table}[ht]
    \centering
    \renewcommand{\arraystretch}{1.3}
     \caption{Iris protocol time in \spdz$\,$ system.}
   \label{tab_iris_spdz}
   
  \begin{tabular}{cccc}
       Iris  &  \multicolumn{2}{c}{CPU time} & Face  \\
       $N$          &    Iris (s)          & multimodal (s)   &  $k$  \\\hline
   
 {1600} &\multirow{2}{*}{0.029s}  &0.030 &1\\
    1600 &               &0.030 &2\\
  3600&  0.048     &0.049 &3\\
  5760& 0.11s& 0.109 &2\\
  6400 &0.12s&0.120 &2
       
     \end{tabular}
   
 \end{table}


%% file: conclusion.tex
In this paper, we have proposed \acronym, a multimodal authentication system based on the MPC 
approach \spdz$\,$ \cite{damgaard2013practical,damgaard2012multiparty} secure against a malicious party.
We have shown that by using a multi-modal system it is possible to improve the 
efficiency of the recognition process in terms of number of multiplications and evaluation time, without any loss of accuracy.
In the same way, it is also possible to improve accuracy at the cost of a negligible 
increase of complexity.
As an additional contribution,  we adapted the iris and face authentication protocols to work in the \spdz 
$\,$ setting. A further additional complexity reduction is achieved by resorting to packed transmission 
of encrypted data involved in the secure multiplication protocol. 
{As future work, we plan to extend our approach to even more biometric traits, like fingerprints,  
behavioral biometric and many others. Another interesting  research direction could be to look for different 
algorithms and more efficient fusion rules to merge the match scores. }
We are also interested to test our protocol on mobile devices, in order to measure the complexity of the 
whole protocol, including also multi-biometric acquisition and feature extraction.
